\documentclass[conference]{IEEEtran}
\usepackage[T1]{fontenc}
\usepackage{lmodern}
\usepackage{textcomp}
\usepackage{latexsym}
\usepackage[fleqn]{amsmath}
\usepackage{bm}
\usepackage{amssymb}
\usepackage[pdftex]{graphicx}
\usepackage{type1cm}
\usepackage{amsthm}
\usepackage{subfigure}
\usepackage{graphicx}

\theoremstyle{definition}

\newtheorem{theorem}{Theorem}

\newtheorem{example}{Example}

\newtheorem{construction}{Construction}

\newcommand{\diag}{{\rm diag}}

\usepackage{setspace} 
\begin{document}
\title{Erasure Correcting Codes by Using Shift Operation and Exclusive OR}
\author{\IEEEauthorblockN{Yuta Hanaki and Takayuki Nozaki}
\IEEEauthorblockA{
Dept.\ of Informatics, Yamaguchi University, JAPAN \\
Email: {\tt \{g022vb,tnozaki\}@yamaguchi-u.ac.jp}}
}

\maketitle
\begin{abstract}
  This paper proposes an erasure correcting code and its systematic form for the distributed storage system.
  The proposed codes are encoded by exclusive OR and bit-level shift operation.
  By the shift operation, the encoded packets are slightly longer than the source packets.
  This paper evaluates the extra length of encoded packets, called overhead, and shows that the proposed codes have smaller overheads than the zigzag decodable code, which is an existing code using exclusive OR and bit-level shift operation.
\end{abstract}

\IEEEpeerreviewmaketitle

\section{Introduction}
The distributed storage systems realize a reliable data storage system via multiple data storage devices.
In the distributed storage systems, each original data (or message) is split into several source packets.
Those source packets are encoded by an erasure correcting code and each encoded packet is stored in a data storage device.
Hence, even if several data storage devices are broken, one can recover the original data by using erasure decoding.
Nowadays, the distributed storage systems are used in cloud storage services \cite{zhang2010cloud}, e.g, Google file system and Dropbox \cite{drago2012inside}.

Each packet is composed of multiple bits.
We assume that an erasure correcting code generates $N$ encoded packets from $K$ source packets, where $N>K$.
An erasure correcting code satisfies {\it combination property} \cite{dai2017new} if the original message can be decoded from arbitrary $K$ source packets.
It is easily confirmed that the maximum distance separative (MDS) codes satisfy combination property.  

Read-Solomon (RS) code \cite{Reed1960polynomial} is an MDS code defined over a non-binary finite field.
Since the encoding and the decoding algorithms of RS codes are performed over the non-binary finite field, the computation complexity is high \cite{vingelmann2010multimedia} and the electric energy consumption is also high \cite{hayinga1999energy}.
Hence, distributed storage systems with RS codes are not suitable for the situation in which one needs high throughput or one needs to save energy (e.g, battery-powered devices).

An erasure correcting code is {\it sub-optimal} if the code satisfies combination property but the length of encoded packets are slightly longer than the source packets.
Zigzag decodable (ZD) code \cite{dai2017new,sung2013zigzag} is a sub-optimal code.
ZD codes are encoded by using exclusive OR (XOR) and bit-level shift operation and are efficiently decoded by zigzag decoding \cite{gollakota2008zigzag}.
Hence, it is known that ZD codes have lower encoding and decoding complexities than RS codes \cite{dai2017new}.

By the bit-level shift operation, the encoded packets are slightly longer than the source packets. 
We refer to the extra length of an encoded packet as {\it overhead} of an encoded packet.
A code with large overhead requires large storage size.
Hence we should construct a code with small overhead.

In this paper, we construct a sub-optimal code which has smaller overhead than the ZD codes.
We refer to the constructed code as the shift and XOR (SXOR) code.
The SXOR code can be also encoded by using XOR and bit-level shift.
In this paper, firstly, we consider a maximum a posteriori (MAP) decoding algorithm for the ZD code.
As a result, we see that the MAP decoding algorithm is efficiently realized by an infinite impulse response (IIR) filter.
Secondly, we construct a sub-optimal code encoded by shift and XOR with small overhead under MAP decoding.
Thirdly, we construct a systematic form of an SXOR code.
Finally, we evaluate the overhead of ZD codes and SXOR codes.
As a result, we show that SXOR codes have smaller overhead than ZD codes.

This paper is organized as follows:
Section \ref{sec:ZigZag} gives notations and definition of ZD codes.
In Section \ref{sec:decoding}, we consider the MAP decoding algorithm for the codes encoded by using XOR and bit-level shift.
Section \ref{sec:pro} proposes SXOR code and its systematic form.
Section \ref{sec:per} evaluates the overhead of ZD codes and SXOR codes.
Section \ref{sec:con} concludes the paper.

\section{ZD Codes and Zigzag Decoding \label{sec:ZigZag}}
This section explains the ZD code and the zigzag decoding algorithm with a toy example.
Moreover, we introduce a construction of a ZD code. 

\begin{example}\label{ex:1}
As a toy example, we consider a ZD code which generates four encoded packets from two source packets with length $4$.
The first (resp. second) encoded packet $\bm{c}_1$ (resp. $\bm{c}_2$) stores the first (resp. second) source packet $\bm{s}_1 = (s_{1,1},s_{1,2},s_{1,3},s_{1,4})$ (resp. $\bm{s}_2 = (s_{2,1},s_{2,2},s_{2,3},s_{2,4})$), i.e, $\bm{c}_1 = \bm{s}_1$ and $\bm{c}_2 = \bm{s}_2$.
The third encoded packet $\bm{c}_3=(c_{3,1},c_{3,2},c_{3,3},c_{3,4})$ is generated from the bit-wise XOR of two source packets $\bm{s}_1,\bm{s}_2$, i.e, $\bm{c}_3=\bm{s}_1+\bm{s}_2$.
The fourth encoded packet $\bm{c}_4=(c_{4,1},c_{4,2},c_{4,3},c_{4,4},c_{4,5})$ is generated from the bit-wise XOR of $\bm{s}_1$ and $\bm{s}_2$ with a right shift, i.e, $\bm{c}_4=(s_{1,1},s_{1,2}+s_{2,1},s_{1,3}+s_{2,2},s_{1,4}+s_{2,3},s_{2,4})$.
Note that the length of the fourth packet is $5$.

Now, consider the decoding from two encoded packets ${\bm c}_3, {\bm c}_4$.
Since the first bit of ${\bm c}_4$ stores the first bit of ${\bm s}_1$,
we have $s_{1,1} = c_{4,1}$.
By using this result, we can recover the first bit of ${\bm s}_2$ from the first bit of ${\bm c}_3$, i.e, $c_{3,1} = s_{2,1} + c_{4,1}$.
Similarly, the decoder recovers $s_{1,2}, s_{2,2}, s_{1,3}, \dots, s_{2,4}$ and the decoding is success.
Since the decoding process takes zigzag path in the encoded packets as in this example, this decoding is called {\it zigzag decoding} \cite{sung2013zigzag}.

\end{example}

We assume that a file is split into $K$ source packets $\bm{s}_1, \bm{s}_2, \cdots, \bm{s}_K$.
Each source packet is composed of $L$ bits.
The $j$-th source packet is denoted by
\begin{equation*}
  \bm{s}_j
   =
  (s_{j,1},s_{j,2},\cdots,s_{j,L}).
\end{equation*}
We introduce the polynomial representation for the source packets easily to describe the shift operation.
The polynomial representation of the source packet $\bm{s}_j$ is given by
\begin{align*}
  s_j(z)
   =
  {\textstyle \sum_{k=1}^{L}} s_{j,k}z^{k-1}.
\end{align*}
A ZD code generates the $N$ encoded packets by using shift operation and XOR of the $K$ source packets.
By the shift operation, the encoded packets are slightly longer than the source packets.
Assuming that the length of the encoded packets $\bm{c}_i$ is given by $L+\ell_i$, we denote
\begin{align*}
  \bm{c}_i
   =
  (c_{i,1},c_{i,2},\cdots,c_{i,L+\ell_i}).
\end{align*}
Here, $\ell_i$ is the number of extra bits generated by the shift operation and  called {\it overhead}.
Similarly, the polynomial representation of the encoded packet is given by
\begin{align}
  c_i(z) 
  = 
  {\textstyle \sum_{j=1}^{L+\ell_i}} c_{i,j}z^{j-1}.
  \label{eq:c}
\end{align}

Each encoded packet is generated as follows: 1) shifting source packets and 2) adding those packets. 
Note that $z^ts_i(z)$ denotes the right shifting of $s_i(z)$ with offset $t$.
Hence, in the ZD code, the $i$-th encoded packet is given as 
\begin{align}
  c_i(z)
  =
  {\textstyle \sum_{j=1}^{K}} a_{i,j}(z)s_j(z), 
  \label{eq:c_sum}
\end{align}
where $a_{i,j}(z)$ is a monomial of $z$, i.e, $a_{i,j}(z) \in \{0,1,z,z^2,\dots\}$.
We denote the degree of $a_{i,j}(z)$, by $\deg(a_{i,j}(z))$.
Then, we have $\ell_i = \max_{1 \leq j \leq K}\deg(a_{i,j}(z))$.
Denote the $K$ source packets and $N$ encoded packets, by
\begin{align*}
  \bm{c}(z)
  :=&
  (c_1(z),c_2(z),\ldots,c_N(z)),
 \\
  \bm{s}(z)
  :=&
  (s_1(z),s_2(z),\ldots,s_K(z)).
\end{align*}
We define the generator matrix by $\mathbf{A}(z) := (a_{i,j}(z))$.
Then, the ZD code is generated as
\begin{align*}
  \bm{c}(z)
   =
   \bm{s}(z)\mathbf{A}(z).
\end{align*}
To simplify the notation, we denote $\mathbf{A}(z)$ by $\mathbf{A}$.

\begin{table}[tb]
  \centering
  \caption{The number of source packets $K$ and maximum overhead $\ell_{max}$ \label{tab:ZD}}
  \begin{small}
  \begin{tabular}{|c|c|c|c|c|} \hline
    $K$ & $2$ & $3$ & $4$ & $5\leq K$ \\ \hline
    $\ell_{\max}$ & $1$ & $1$ & $3$ & $K(K-1)/2$ \\ \hline
  \end{tabular}
  \end{small}
\end{table}
The ``good'' ZD code can be decoded by zigzag decoding and has the small {\it maximum overhead} $\ell_{\max} := \mathrm{max}_{1 \leq i \leq N}\ell_i$ and {\it total overhead} $\ell_{\mathrm{sum}} := \sum_{i=1}^{N}\ell_i$. 
In \cite{ZigZag3}, \cite{ZigZag1}, the ZD codes with $2K=N$ are proposed.
The ZD codes with the smallest maximum overhead is given in \cite{ZigZag3}.
Table \ref{tab:ZD} shows those maximum overhead.
In \cite{ZigZag3}, for $K=2,3,4$, the generator matrixes via heuristic approach and for $K \geq 5$ the generator matrixes are constructed from Hankel matrixes.
For example, the generator matrix with $K=3$ in \cite{ZigZag3} is
\begin{align}
  \mathbf{A}
  =
  \begin{pmatrix}
    1 & 0 & 0 & 1 & z & z \\
    0 & 1 & 0 & z & 1 & z \\
    0 & 0 & 1 & z & z & 1 
  \end{pmatrix}. \label{eq:K=3}
\end{align} 
\section{MAP Decoding Algorithm for ZD Code} \label{sec:decoding}
In this section, we will show that the ZD codes are also efficiently decoded by MAP decoding algorithm.

Let $\mathbb{F}_2$ be the finite field of order $2$.
Let $\mathbb{F}_2[z]$ be the polynomial ring with the coefficient $\mathbb{F}_2$.
Moreover, we denote field of rational functions over $\mathbb{F}_2$ as $\mathbb{F}_2(z)$, i.e,
\begin{equation*}
 \mathbb{F}_2(z)
  :=
 \{ f(z)/g(z) \mid 
  f(z),g(z)\in\mathbb{F}_2[z], g(z) \not= 0 \}.
\end{equation*}
In the MAP decoding algorithm, the source packets are decoded from $K$ encoded packets.
We denote the $K$ encoded packets, by $c_{i_1},c_{i_2},\dots,c_{i_K}$.
Let $\mathcal{I}$ be the set of indexes of the encoded packets, i.e,  $\mathcal{I} := \{ i_1, i_2, \dots, i_K \}$. 
We denote the $K\times K$ submatrix of $\mathbf{A}$ obtained by choosing columns in the set $\mathcal{I}$, by $\mathbf{A}_{\mathcal{I}}$.
Notice that
\begin{align} 
  (s_{1},s_{2},\ldots,s_{K})
  \mathbf{A}_{\mathcal{I}}
  =
  (c_{i_1},c_{i_2},\ldots,c_{i_K}). \label{eq:submatrix}
\end{align}
If $\mathbf{A}_{\mathcal{I}}$ is the invertible matrix over $\mathbb{F}_2(z)$, we have
\begin{align*}
  (s_{1},s_{2},\ldots,s_{K})
  = 
  (c_{i_1},c_{i_2},\cdots,c_{i_K})
  \mathbf{A}_{\mathcal{I}}^{-1}.
\end{align*}

\begin{example}\label{ex:2}
  We assume that the generator matrix $\mathbf{A}$ is \eqref{eq:K=3}.
When we decode the source packets from the encoded packets $c_4(z),c_5(z),c_6(z)$, i.e, $\mathcal{I} = \{4,5,6\}$, we have
\begin{align*}
  \mathbf{A}_{\mathcal{I}}
  =
  \begin{pmatrix}
    1 & z & z \\
    z & 1 & z \\
    z & z & 1
  \end{pmatrix}.
\end{align*}
Then, the inverse matrix is
\begin{align}\label{eq:reverse}
  \mathbf{A}_{\mathcal{I}}^{-1}
  =
  \frac{1}{z+1}
  \begin{pmatrix}
    z+1 & z & z \\
    z & z+1 & z \\
    z & z & z+1 \\
  \end{pmatrix} 
  =:
  \frac{1}{z+1} \mathbf{B}.
\end{align}
Therefore, firstly, we calculate the following equation.
\begin{align*}
  \begin{pmatrix}
    b_1(z) & b_2(z) & b_3(z)
  \end{pmatrix}
  :=
  \begin{pmatrix}
      c_4(z) & c_5(z) & c_6(z)
    \end{pmatrix}
    \mathbf{B}.
  \end{align*}
From \eqref{eq:submatrix}, we get  
  \begin{align*}
    (c_4(z), c_5(z), c_6(z))
    =  (s_1(z), s_2(z), s_3(z))\mathbf{A}_{\mathcal{I}}.
  \end{align*}
  From those equations, we have 
  \begin{equation*}
    (b_1(z), b_2(z), b_3(z)) = (z+1) (s_1(z), s_2(z), s_3(z)).
  \end{equation*}

Secondly, we calculate $(s_1(z),s_2(z),\allowbreak s_3(z))$ from $(b_1(z),b_2(z),b_3(z))$.
Since $b_i(z) = (z+1)s_i(z)$, 
  \begin{equation}
    b_{i,t+1} = s_{i,t+1} + s_{i,t}, \label{eq:ex_const}
  \end{equation}
where $s_{i,0}=0$.
In particular, we have $b_{i,1} = s_{i,1}$, and we recover $s_{i,1}$.
Next, we have $b_{i,2}$ and $s_{i,1}$, and we recover $s_{i,2}$.
Similarly, if decoding succeeds until the $j$-th bit, we can compute the $(j+1)$-th bit by substituting $t=j$ into \eqref{eq:ex_const}.
This can be easily realized by feedback and a flip-flop.

Figure \ref{fig:kairo} depicts the decoding circuit for the example.
The boxes with $z$ in Fig.~\ref{fig:kairo} represent flip-flops.
The left side of Fig.~\ref{fig:kairo} computes of $z c_i(z)$ and $(z+1) c_i(z)$.
The middle of Fig.~\ref{fig:kairo} computes $b_1(z),b_2(z),b_3(z)$.
The right side of Fig.~\ref{fig:kairo} derives $s_1(z),s_2(z),s_3(z)$.
  
  \begin{figure}[tb]
    \centering
    \includegraphics[scale=0.3]{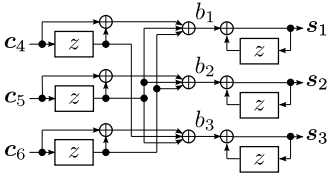}
    \caption{The circuit of the MAP decoding for Example \ref{ex:2}}
    \label{fig:kairo}
  \end{figure}
  
\end{example}
For a fixed ZD code and a set of indexes, if Zigzag decoding succeeds, the MAP decoding also succeeds.
Hence, if ZD code is sub-optimal, the MAP decoding algorithm can decode the source packets from arbitrary $K$ encoded packets.
That is, if a ZD code is sub-optimal, we have $\det \mathbf{A}_{\mathcal{I}} \not= 0$ for arbitrary set of indexes $\mathcal{I}$ with $|\mathcal{I}|=K$.
Conversely, if the submatrix $\mathbf{A}_{\mathcal{I}}$ satisfies $\det \mathbf{A}_{\mathcal{I}} \not= 0$ for arbitrary set of indexes $\mathcal{I}$ with $|\mathcal{I}|=K$, then the ZD code is sub-optimal under MAP decoding.

In general, we can rewrite $\mathbf{A}_{\mathcal{I}}^{-1}$ with $h(z) \in \mathbb{F}_2[z]$ and a matrix $\mathbf{B}$ over $\mathbb{F}_2[z]$, i.e,
\begin{equation*}
 \mathbf{A}_{\mathcal{I}}^{-1}
  =
 \frac{1}{h(z)} \mathbf{B}.
\end{equation*}
The polynomial $h(z)$ can be factorized by a monomial $z^t$ and irreducible polynomials with  constant term $1$ as follows:
\begin{equation*}
 h(z) = z^t h_1(z) h_2(z) \cdots h_s(z).
\end{equation*}
This $z^t$ means that the $t$ bits at the front of $s_i(z)$ are $0$.
Therefore, the $t$ bits at the front of $s_i(z)$ are removed in the decoding.
The polynomials $h_1(z),h_2(z),\dots, h_s(z)$ are realized by feed back and flip-flops.
Moreover, $h(z)$ is constructed from the cascade of those filters.

The MAP decoding algorithm depicted in Fig.\ref{fig:kairo} starts from the left of the encoded packets.
Hence, the source packets can be decoded from the left.

The zigzag decoding algorithm is to search an encoded packet that has an exposed bit, which can be directly read out.
After that, the bit is subtracted from other encoded packets.
The procedure repeats until all source packets are decoded.
In other words, the bits of $\bm{s}_1$ and $\bm{s}_2$ {\it cannot} be decoded in parallel under zigzag decoding.
The MAP decoding algorithm by using a circuit can decode $\bm{s}_1,\bm{s}_2,\dots,\bm{s}_K$ in parallel, i.e, the MAP decoding algorithm is the parallel decoding algorithm.

\section{Erasure Correcting Codes by Using Shift Operation and Exclusive OR} \label{sec:pro}

In this section, we propose the erasure correcting codes by using shift operation and XOR which is named the shift and XOR (SXOR) code.
Moreover, we propose a systematic SXOR code.

\subsection{SXOR codes}
Let $\mathbb{F}_{2^m}$ be the finite field of order $2^m$.
Let $z$ be a primitive element of $\mathbb{F}_{2^m}$ and $g(z)$ a primitive polynomial of which root is $z$. 
Assume that $K\leq N \leq 2^m-1$, where $m$ is a positive integer.
Let $\mathbf{V} = (v_{i,j})$ be a $K\times N$ Vandermonde matrix defined over $\mathbb{F}_{2^m}$, i.e, $v_{i,j} = z^{(i-1)(j-1)}$.
We denote the remainder derived from division of polynomial $a(x)$ by $g(z)$, by $\langle a \rangle$.

\begin{construction}[SXOR code] \label{con:1}
  Each element of a generator matrix for an SXOR code is a polynomial remainder derived from the division of the corresponding element of $\mathbf{V}$ by $g(z)$.
  In other words, the $K \times N$ generator matrix $\mathbf{A} = (a_{i,j})$ for the SXOR code satisfies $a_{i,j} = \langle z^{(i-1)(j-1)} \rangle$.
\end{construction}

Notice that the maximum overhead is determined from the maximum degree of the elements in generator matrix.
Hence, the maximum overhead is $m-1$.
Recall that $m=\lceil \log_2 (N+1) \rceil$, where the symbol $\lceil a \rceil$ denotes the ceiling function, which is the integer obtained by rounding up.
Therefore, the maximum overhead $\ell_{\max}$ is given by $\lceil \log_2(N+1) - 1 \rceil$.

\begin{example}\label{ex:4}
  Let $K=3$, $ N=7$ and $g(z)=z^3+z+1$.
  The generator matrix $\mathbf{A}$ is given by
    \begin{align*}
  \begin{small}
      \mathbf{A}
      \!=\!
      \begin{pmatrix}
        \!1\! & \!1\! & \!1\! & \!1\! & \!1\! & \!1\! & \!1\! \\
        \!1\! & \!z\! & \!z^2\! & \!z\!+\!1\! & \!z^2\!+\!z\! & \!z^2\!+\!z\!+\!1\! & \!z^2\!+\!1\! \\
        \!1\! & \!z^2\! & \!z^2\!+\!z\! & \!z^2\!+\!1\! & \!z\! & \!z\!+\!1\! & \!z^2\!+\!z\!+\!1\!
      \end{pmatrix}.
  \end{small}
    \end{align*}
From this, we see that the maximum degree, i.e, the maximum overhead, is $2$.
\end{example}

Hereafter, we denote the $K\times K$ submatrix of $\mathbf{V}$ obtained by choosing columns in the set $\mathcal{I}$, by $\mathbf{V}_{\mathcal{I}}$.
\begin{theorem} \label{the:0}
 The code in Construction \ref{con:1} is sub-optimal.
\end{theorem}

\begin{proof}
To prove this, we will show that the submatrix $\mathbf{A}_{\mathcal{I}}$ for the arbitrary set of indexes $\mathcal{I}$ with $|\mathcal{I}|=K$ satisfies $\det \mathbf{A}_{\mathcal{I}} \not = 0$. 
By the properties of Vandermonde matrix, we have
\begin{align*}
  \det \mathbf{V}_{\mathcal{I}} 
  = 
  {\textstyle \prod_{1\leq a < b \leq K}} \bigl(z^{i_b-1}-z^{i_a-1}\bigr).
\end{align*}
Note that primitive polynomials are irreducible.
By the properties of Vandermonde matrix and remainder, we have
\begin{align*}
  \langle \det \mathbf{A}_{\mathcal{I}} \rangle 
  = 
  \left\langle {\textstyle \prod_{1\leq a < b \leq K}}
  \bigr(z^{i_b-1}-z^{i_a-1}\bigr) \right\rangle 
  = 
  r(z),
\end{align*}
where $r(z)$ be a nonzero polynomial of degree less than $m$.
From the above, we can write with a polynomial $d(z)\in \mathbb{F}_2[z]$
\begin{align*}
\det \mathbf{A}_{\mathcal{I}} = d(z)g(z)+r(z) \neq 0.
\end{align*}
Hence we obtain the theorem.
\end{proof}

\subsection{Systematic SXOR Codes \label{sec:sys}}
Let $\bm{x} = (x_1, x_2, \dots, x_K)$ be a sequence over $\{1, 2, \dots, N\}$ which satisfies $x_i \not = x_j$ (for $i \not = j$).
Let $\mathbf{V}_{\bm{x}}$ be the $K\times K$ submatrix of which the $i$-th column equals to the $x_i$-th column of $\mathbf{V}$.
\begin{example}
  Let $K=3$ and $N=7$.
  For the sequence $\bm{x}=(1,3,4)$ and $\bm{y}=(1,4,3)$, we have
   \begin{align*}
    \mathbf{V}_{\bm{x}}=
    \begin{pmatrix}
      1 & 1 & 1 \\
      1 & z^2 & z^3 \\
      1 & z^4 & z^6 
    \end{pmatrix},    
    \mathbf{V}_{\bm{y}}=
    \begin{pmatrix}
      1 & 1 & 1 \\
      1 & z^3 & z^2 \\
      1 & z^6 & z^4 
    \end{pmatrix}.    
  \end{align*}
\end{example}

By the properties of Vandermonde matrix, the inverse matrix $\mathbf{V}_{\bm{x}}^{-1}$ always exists.

\begin{construction}[Systematic SXOR code] \label{con:2}
  Define
  \begin{align} \label{eq:tildeA}
    \mathbf{\tilde{A}}_{\bm{x}}
    &:=\mathbf{V}_{\bm{x}}^{-1}\mathbf{V}=
      \begin{pmatrix}
        \tilde{a}_{i,j}(z) 
      \end{pmatrix}
  \end{align}
  Similar to the Construction \ref{con:1}, each element of a generator matrix is a polynomial remainder derived from the division of corresponding element of $\mathbf{\tilde{A}}_{\bm{x}}$ by $g(z)$.
  In other words, the systematic SXOR code is generated from an $K\times N$ generator matrix $\mathbf{A}_{\bm{x}} = (\langle \tilde{a}_{i,j}(z) \rangle)$.
\end{construction}

\begin{example}
  Let $K=3$, $N=7$, $g(z)=z^3+z+1$ and $\bm{x}=(1,3,4)$.
  The submatrix $\mathbf{V}_{\bm{x}}$ and inverse matrix $\mathbf{V}_{\bm{x}}^{-1}$ is given by
  \begin{align*}
    \mathbf{V}_{\bm{x}}^{-1}=
    \begin{pmatrix}
      z^5 & z^5 & 1 \\
      z^6 & z^4 & z^3 \\
      z^3 & 1 & z 
    \end{pmatrix}.    
  \end{align*}
Compute $\mathbf{\tilde{A}}_{\bm{x}}=\mathbf{V}_{\bm{x}}^{-1}\mathbf{V}$.
Then, the generator matrix $\mathbf{A}_{\bm{x}}$ is
  \begin{align*}
    \mathbf{A}_{\bm{x}}
    \!=\!
    \begin{pmatrix}
    1 \!&\! z^2+z \!&\! 0 \!&\! 0 \!&\! 1 \!&\! z^2+z+1 \!&\! z^2+z \\
    0 \!&\! z^2+1 \!&\! 1 \!&\! 0 \!&\! 1 \!&\! z^2 \!&\! z^2 \\
    0 \!&\! z \!&\! 0 \!&\! 1 \!&\! 1 \!&\! z \!&\! z+1 
    \end{pmatrix}.
  \end{align*}
  From the first, third, and fourth column of $\mathbf{A}_{\bm{x}}$, we see that the encoded packets store the source packets.
  This means that the code is systematic.
\end{example}

\begin{theorem} \label{the:1}
   The code in Construction \ref{con:2} is sub-optimal.
\end{theorem}
This theorem is proven in a similar way to Theorem \ref{the:0}.


The generator matrix of a systematic SXOR code depends on $\bm{x}$.
Hence, we see that there are $\frac{N!}{(N-K)!}$ generator matrixes.
However, those matrixes can be classified into several classes.

A sequence $\bm{x} = (x_1, s_2, \cdots, s_K)$ is {\it equivalent} to $\bm{y} = (y_1, y_2, \cdots, y_K)$ if for all $i$ there exists only one $j$ such that $x_i=y_j$.
In other words, $\bm{x}$ and $\bm{y}$ are equivalent if we can write $x_i = y_{\sigma(i)}$ by a permutation $\sigma$ over $\{1, 2, \cdots, K\}$.
The two generator matrixes $\mathbf{A}_{\bm{x}}$ and $\mathbf{A}_{\bm{y}}$ are {\it equivalent} if $\mathbf{A}_{\bm{x}}$ can be transformed into $\mathbf{A}_{\bm{y}}$ by using row permutation and column permutation.
\begin{theorem} \label{the:2}
  If $\bm{x}$ and $\bm{y}$ are equivalent, then $\mathbf{A}_{\bm{x}}$ and $\mathbf{A}_{\bm{y}}$ are equivalent.
\end{theorem}

\begin{proof}
We denote the submatrix $\mathbf{V}_{\bm{x}}$ and $\mathbf{V}_{\bm{x}}^{-1}$, by
$\mathbf{V}_{\bm{x}} 
 = 
 \begin{pmatrix} \bar{v}_{i,j} \end{pmatrix}$, 
$\mathbf{V}_{\bm{x}}^{-1}=
  \begin{pmatrix} \hat{v}_{i,j} \end{pmatrix}$.
Since $\bm{x}$ and $\bm{y}$ are equivalent, we have
$\mathbf{V}_{\bm{y}}=
  \begin{pmatrix}  \bar{v}_{\sigma(i),j} \end{pmatrix}$,
$\mathbf{V}_{\bm{y}}^{-1}=
  \begin{pmatrix} \hat{v}_{i,\sigma(j)} \end{pmatrix}$.
Note that $\mathbf{V}_{\bm{x}}^{-1}$ and $\mathbf{V}_{\bm{y}}^{-1}$ are equivalent.
From
$\mathbf{V}_{\bm{x}}^{-1}\mathbf{V}=\mathbf{\tilde{A}}_{\bm{x}}$ and
$\mathbf{V}_{\bm{y}}^{-1}\mathbf{V}=\mathbf{\tilde{A}}_{\bm{y}}$,
$\mathbf{\tilde{A}}_{\bm{x}}$ and $\mathbf{\tilde{A}}_{\bm{y}}$ are equivalent.
From Construction \ref{con:2}, $\mathbf{A}_{\bm{x}}$ and $\mathbf{A}_{\bm{y}}$ are equivalent.
\end{proof}

From Theorem \ref{the:2}, the generator matrixes can be classified into $\frac{N!}{(N-K)!K!}=\binom{N}{K}$ classes.

\begin{example} 
  Let $K=3$, $N=7$ and $g(z)=z^3+z+1$.
  We will enumerate the generator matrix $\mathbf{A}_{\bm{x}}$ with sequence with entries $1$, $3$ and $4$.
To simplify the notation, we denote the first, second, and third row of $\mathbf{A}_{(1,3,4)}$, by $\bar{\bm{a}}_1$, $\bar{\bm{a}}_3$, and $\bar{\bm{a}}_4$, respectively.
  Then, the row vector $\bar{\bm{a}}_1$, $\bar{\bm{a}}_3$ and $\bar{\bm{a}}_4$ are
  \begin{align*}
    \bar{\bm{a}}_1
    &=
    \begin{pmatrix}
      1 & z^2+z & 0 & 0 & 1 & z^2+z+1 & z^2+z  
    \end{pmatrix}, \\
    \bar{\bm{a}}_3
    &=
    \begin{pmatrix}
      0 & z^2+1 & 1 & 0 & 1 & z^2 & z^2 
    \end{pmatrix}, \\
    \bar{\bm{a}}_4
    &=
    \begin{pmatrix}
      0 & z & 0 & 1 & 1 & z & z+1
    \end{pmatrix}.
  \end{align*}
  Then, we obtain the following results:
  \begin{align*}
    \mathbf{A}_{(1,3,4)}=
    \begin{pmatrix}
      \bar{\bm{a}}_1 \\ \bar{\bm{a}}_3 \\ \bar{\bm{a}}_4
    \end{pmatrix},
    \mathbf{A}_{(1,4,3)}=
    \begin{pmatrix}
      \bar{\bm{a}}_1 \\ \bar{\bm{a}}_4 \\ \bar{\bm{a}}_3
    \end{pmatrix},
    \mathbf{A}_{(4,1,3)}=
    \begin{pmatrix}
      \bar{\bm{a}}_4 \\ \bar{\bm{a}}_1 \\ \bar{\bm{a}}_3
    \end{pmatrix}, \\
    \mathbf{A}_{(4,3,1)}=
    \begin{pmatrix}
      \bar{\bm{a}}_4 \\ \bar{\bm{a}}_3 \\ \bar{\bm{a}}_1
    \end{pmatrix},
      \mathbf{A}_{(3,1,4)}=
    \begin{pmatrix}
      \bar{\bm{a}}_3 \\ \bar{\bm{a}}_1 \\ \bar{\bm{a}}_4
    \end{pmatrix},
    \mathbf{A}_{(3,4,1)}=
    \begin{pmatrix}
      \bar{\bm{a}}_3 \\ \bar{\bm{a}}_4 \\ \bar{\bm{a}}_1
    \end{pmatrix}.
  \end{align*}
  From the above, we confirm that each entry of $\bm{x}$ corresponds to the index of row vector.
\end{example}

In particular, for $N=2^m-1$, we can reduce the number of classes of the generator matrixes.
For ${\bm x} = (x_1,x_{2},\dots,x_{K})$ and a positive interger $k$ $(1 \leq k < N)$, we define ${\bm y} = \bm{x} + \bm{k} = (y_1, y_2,\dots, y_K)$ as follows:
\begin{equation*}
 y_i 
 =
 \begin{cases}
  x_i + k, & \text{if}~ x_i+k \le N, \\
  x_i + k - N, & \text{otherwise}.
 \end{cases}
\end{equation*}

\begin{theorem} \label{the:3}
  Assume that $N=2^m-1$.
  Then, $\mathbf{A}_{\bm{x}}$ and $\mathbf{A}_{\bm{x}+\bm{k}}$ are equivalent for all positive integer $k$.
\end{theorem}

\begin{proof}
Write $\diag(1, z, z^2 ,\cdots, z^{K-1})$ for a diagonal matrix whose diagonal entries starting in the upper left corner are $1, z, z^2 ,\cdots, z^{K-1}$. 
From the properties of the Vandermonde matrix, we have for $k=1$
\begin{align} \label{eq:the:3}
  \mathbf{V}_{\bm{x}+\bm{1}}=
  \diag(1, z, z^2 ,\cdots, z^{K-1})
  \mathbf{V}_{\bm{x}}.
\end{align}
From \eqref{eq:tildeA} and \eqref{eq:the:3}, the inverse matrix $\mathbf{V}_{\bm{x}+\bm{1}}^{-1}$ is given by
\begin{align*}
  \mathbf{V}_{\bm{x}+\bm{1}}^{-1}=
  \mathbf{V}_{\bm{x}}^{-1}
  \diag(1, z^{N-1}, z^{N-2} ,\cdots, z^{N-K+1}).
\end{align*}
From the above and \eqref{eq:tildeA}, we have
{\footnotesize
  \begin{align*}
    \mathbf{\tilde{A}}_{\bm{x}+\bm{1}}
    \!\! &= \!\!
    \mathbf{V}_{\bm{x}}^{-1} 
    \begin{pmatrix}
      \!1\! \!&\! \!0\! \!&\! \cdots \!&\! \!0\! \\
      \!0\! \!&\! \!z\! \!&\! \cdots \!&\! \!0\! \\
      \!\vdots\! \!&\! \!\vdots\! \!&\! \!\ddots\! \!&\! \!\vdots\! \\
      \!0\! \!&\! \!0\! \!&\! \!\cdots\! \!&\! \!z^{K-1}\!
    \end{pmatrix}^{\!\!-\!1}\!\!
    \begin{pmatrix}
      \!1\! \!&\! \!1\! \!&\! \!\cdots\! \!&\! \!1\! \\
      \!1\! \!&\! \!z\! \!&\! \!\cdots\! \!&\! \!z^{(N-1)}\! \\
      \!\vdots\! \!&\! \!\vdots\! \!&\! \!\vdots\! \!&\! \!\vdots\! \\
      \!1\! \!&\! \!\! z^{(K-1)} \!\! \!&\! \cdots \!&\! \!\! z^{(N-1)(K-1)} \!\!
    \end{pmatrix} \!\!\!\\
    &=
    \mathbf{V}_{\bm{x}}^{-1}
    \begin{pmatrix}
      \!1\! \!&\! \!1\! \!&\! \!1\! \!&\! \!\cdots\! \!&\! \!1\! \!&\! \!1\! \\
      \!z^{(\!N\!-\!1\!)}\! \!&\! \!1\! \!&\! \!z\! \!&\! \!\cdots\! \!&\! \!z^{(\!N\!-\!3\!)}\! \!&\! \!z^{(\!N\!-\!2\!)}\! \\
      \!z^{2(\!N\!-\!1\!)}\! \!&\! \!1\! \!&\! \!z^2\! \!&\! \!\cdots\! \!&\! \!z^{2(\!N\!-\!3\!)}\! \!&\! \!z^{2(\!N\!-\!2\!)}\! \\
      \!\vdots\! \!&\! \!\vdots\! \!&\! \!\vdots\! \!&\! \!\vdots\! \!&\! \!\vdots\! \!&\! \!\vdots\! \\
      \!z^{(\!N\!-\!1\!)(\!K\!-\!1\!)}\! \!&\! \!1\! \!&\! \!z^{(\!K\!-\!1\!)}\! \!&\! \!\cdots\! \!&\! \!z^{(\!N\!-\!3\!)(\!K\!-\!1\!)}\! \!&\! \!z^{(\!N\!-\!2\!)(\!K\!-\!1\!)}\!
    \end{pmatrix}.
\end{align*}}
From this, the second factor of the equation is equivalent to $\mathbf{V}$.
Hence, $\mathbf{\tilde{A}}_{\bm{x}}$ is equivalent to $\mathbf{\tilde{A}}_{\bm{x}+\bm{1}}$.
By similar arguments to the above, we can also show that for $\bm{x}, \bm{x}+\bm{1}, \cdots, \bm{x}+\bm{k}$.
Hence we obtain the theorem.
\end{proof}

\begin{example}
  Let $K=3$, $N=7$, $\bm{x}=(1,3,4)$ and $g(z)=z^3+z+1$.
  We will enumerate the generator matrix $\mathbf{A}_{\bm{x}}$ with sequence with $\bm{x}, \bm{x}+\bm{1}, \cdots, \bm{x}+\bm{k}$.
  Denote the $i$-th column of $\mathbf{A}_{\bm{x}}$, by $\bm{a}_i$.
    We obtain the following results:
    \begin{align*}
      \mathbf{A}_{\bm{x}+\bm{1}}
      &=
      \begin{pmatrix}
        \bm{a}_7 & \bm{a}_1 & \bm{a}_2 & \bm{a}_3 & \bm{a}_4 & \bm{a}_5 & \bm{a}_6 
      \end{pmatrix}, \\
      \mathbf{A}_{\bm{x}+\bm{2}}
      &=
      \begin{pmatrix}
        \bm{a}_6 & \bm{a}_7 & \bm{a}_1 & \bm{a}_2 & \bm{a}_3 & \bm{a}_4 & \bm{a}_5 
      \end{pmatrix}, \\
      &\hspace{10mm}\vdots  \\
      \mathbf{A}_{\bm{x}+\bm{6}}
      &=
      \begin{pmatrix}
        \bm{a}_2 & \bm{a}_3 & \bm{a}_4 & \bm{a}_5 & \bm{a}_6 & \bm{a}_7 & \bm{a}_1 
      \end{pmatrix}.
    \end{align*}
    From the above, we confirm that $\mathbf{A}_{\bm{x}+\bm{k}}$ is the $k$ right cyclic shift of $\mathbf{A}_{\bm{x}}$.
\end{example}

\section{Performance Evaluation \label{sec:per}}

In this section, we evaluate the encoding complexity and overhead for SXOR codes.
Section \ref{subsec:ex} shows that the complexity and overhead are depended on the sequence $\bm{x}$ and primitive polynomial $g(z)$.
Section \ref{subsec:compar} compares the complexity and overhead for the SXOR codes, systematic SXOR codes and Hankel matrix based ZD codes \cite{ZigZag3}.

\subsection{Dependency to $\bm{x}$ and $g(z)$ \label{subsec:ex}}

We denote the two primitive polynomials $g(z)$ of degree $3$, by $g_1(z)=z^3+z+1$ and $g_2(z)=z^3+z^2+1$.
We assume $K=3$, $N=7$.
In this case, from the result in the previous section, the generator matrixes are classified into five classes.
We denote the five representatives, by $\mathbf{A}_{(1,2,3)}$, $\mathbf{A}_{(1,2,4)}$, $\mathbf{A}_{(1,2,5)}$, $\mathbf{A}_{(1,3,4)}$, and $\mathbf{A}_{(1,2,5)}$. 
Then, the five representatives for $g_1(z)$ are
  \begin{align*}
    &\begin{footnotesize}
    \mathbf{A}_{(1,2,3)}
    \!\! = \!\!
           \begin{pmatrix}
             \! 1 \! & \! 0 \! & \! 0 \! & \! z\!+\!1 \! & \! z \! & \! 1 \! &  \! z\!+\!1 \! \\
             \! 0 \! & \! 1 \! & \! 0 \! & \! z^2\!+\!1 \! & \! z^2\!+\!1 \! & \! 1 \! & \! z^2 \! \\
             \! 0 \! & \! 0 \! & \! 1 \! & \! z^2\!+\!z\!+\!1 \! & \! z^2\!+\!z \! & \! 1 \! & \! z^2\!+\!z \! \\
           \end{pmatrix}\!\!,
    \end{footnotesize}\\
    &\begin{footnotesize}
    \mathbf{A}_{(1,2,4)}
    \!\! = \!\!
           \begin{pmatrix}
             \! 1 \! & \! 0 \! & \! z^2\!+\!z\!+\!1 \! & \! 0 \! & \! z^2\!+\!z \! & \! z^2\!+\!z \! & \! z^2\!+\!z\!+\!1 \! \\
             \! 0 \! & \! 1 \! & \! z \! & \! 0 \! & \! z \! & \! z\!+\!1 \! & \! z\!+\!1 \! \\
             \! 0 \! & \! 0 \! & \! z^2 \! & \! 1 \! & \! z^2\!+\!1 \! & \! z^2 \! & \! z^2\!+\!1 \! \\
           \end{pmatrix}\!\!,
    \end{footnotesize}\\
    &\begin{footnotesize}
    \mathbf{A}_{(1,2,5)}
    \!\! = \!\!
           \begin{pmatrix}
             \! 1 \! & \! 0 \! & \! z^2\!+\!z \! & \! z^2\!+\!z\!+\!1 \! & \! 0 \! & \! z^2\!+\!z\!+\!1 \! & \! 1 \! \\
             \! 0 \! & \! 1 \! & \! z^2 \! & \! z^2 \! & \! 0 \! & \! z^2\!+\!1 \! & \! 1 \! \\
             \! 0 \! & \! 0 \! & \! z+1 \! & \! z \! & \! 1 \! & \! z+1 \! & \! 1 \! \\ 
           \end{pmatrix}\!\!,
    \end{footnotesize}\\
    &\begin{footnotesize}
    \mathbf{A}_{(1,3,4)}
    \!\! = \!\!
           \begin{pmatrix}
             \! 1 \! & \! z^2\!+\!z  \! & \! 0 \! & \! 0 \! & \! 1 \! & \! z^2\!+\!z\!+\!1 & \! z^2\!+\!z  \\
             \! 0 \! & \! z^2\!+\!1 \! & \! 1 \! & \! 0 \! & \! 1 \! & \! z^2 \! & \! z^2 \! \\
             \! 0 \! & \! z \! & \! 0 \! & \! 1 \! & \! 1 \! & \! z \! & \! z\!+\!1 \! \\
           \end{pmatrix}\!\!,
    \end{footnotesize}\\
    &\begin{footnotesize}
    \mathbf{A}_{(1,3,5)}
    \!\! = \!\!
           \begin{pmatrix}
             \! 1 \! & \! z^2 \! & \! 0 \! & \! 1 \! & \! 0 \! & \! z^2\!+\!1 \! & \! z^2\!+\!1 \! \\
             \! 0 \! & \! z^2\!+\!z\!+\!1 \! & \! 1 \! & \! 1 \! & \! 0 \! & \! z^2\!+\!z \! & \! z^2\!+\!z\!+\!1 \! \\
             \! 0 \! & \! z \! & \! 0 \! & \! 1 \! & \! 1 \! & \! z \! & \! z\!+\!1 \! \\
           \end{pmatrix}\!\!.
         \end{footnotesize}
  \end{align*}
Due to space limitations, we omit the generator matrixes for $g_2(z)$.

We refer to the number of XOR used in encoder as the encoding complexity and denote it by $\alpha$.
Table \ref{tab:J} displays the maximum overhead $\ell_{\max}$, total overhead $\ell_{\mathrm{sum}}$, and encoding complexity $\alpha$ of the five representatives for $g_1(z)$ and $g_2(z)$.
\begin{table}[t]
  \centering
  \caption{The maximum overhead, total overhead, and encoding complexity of five representatives for $g_1(z)$ and $g_2(z)$ \label{tab:J}} 
  \begin{small}
  \begin{tabular}{|c|c|c|c||c|c|c|} \hline
    & \multicolumn{3}{|c||}{$g_1(z)$} & \multicolumn{3}{|c|}{$g_2(z)$}  \\ \hline
    & $\ell_{\max}$ & $\ell_{\mathrm{sum}}$ & $\alpha$ & $\ell_{\max}$ & $\ell_{\mathrm{sum}}$ & $\alpha$ \\ \hline\hline
    $\mathbf{A}_{(1,2,3)}$ & $2$ & $6$ & $16$ & $2$ & $6$ & $16$ \\ \hline
    $\mathbf{A}_{(1,2,4)}$ & $2$ & $8$ & $18$ & $2$ & $6$ & $16$ \\ \hline
    $\mathbf{A}_{(1,2,5)}$ & $2$ & $6$ & $16$ & $2$ & $6$ & $16$ \\ \hline
    $\mathbf{A}_{(1,3,4)}$ & $2$ & $6$ & $14$ & $2$ & $8$ & $18$ \\ \hline
    $\mathbf{A}_{(1,3,5)}$ & $2$ & $6$ & $16$ & $2$ & $6$ & $14$ \\ \hline
  \end{tabular}
  \end{small}
\end{table}
From Table \ref{tab:J}, we see that the total overhead and encoding complexity depend on the sequence $\bm{x}$ a primitive polynomial $g(z)$.

\subsection{Performance comparison \label{subsec:compar}}

We assume $N=7$ and $g(z)=z^3+z+1$.
The complexities and overheads for systematic SXOR codes depend on the sequence $\bm{x}$ for each $K$.
Hence, we evaluate the systematic SXOR code with the smallest total overhead.
We compare the systematic SXOR code with the SXOR code and the ZD code.
Table \ref{tab:1} shows the encoding complexity $\alpha$ and two overheads $\ell_{\max}$, $\ell_{\mathrm{sum}}$ for systematic SXOR codes, SXOR codes and ZD codes.

\begin{table}[t]
  \centering
  \caption{The encoding complexity and two overheads of systematic SXOR codes, SXOR codes and ZD codes \label{tab:1}} 
  \begin{small}
    \begin{tabular}{|c|c|c|c|c|c|c|} \hline
      & $K$ & $2$ & $3$ & $4$ & $5$ & $6$ \\ \hline \hline
      & $\ell_{\max}$ & $2$ & $2$ & $2$ & $2$ & $2$ \\ \cline{2-7}
      Systematic SXOR code & $\ell_{\mathrm{sum}}$ & $8$ & $6$ & $6$ & $3$ & $2$ \\ \cline{2-7}
      & $\alpha$ & $12$ & $14$ & $12$ & $12$ & $10$ \\ \hline
      & $\ell_{\max}$ & $2$ & $2$ & $2$ & $2$ & $2$ \\ \cline{2-7}
      SXOR code & $\ell_{\mathrm{sum}}$ & $10$ & $11$ & $12$ & $12$ & $12$ \\ \cline{2-7}
      & $\alpha$ & $12$ & $24$ & $36$ & $48$ & $60$ \\ \hline
      & $\ell_{\max}$ & $3$ & $3$ & $3$ & $3$ & $3$ \\ \cline{2-7}
      ZD code & $\ell_{\mathrm{sum}}$ & $8$ & $8$ & $7$ & $6$ & $3$ \\ \cline{2-7}
      & $\alpha$ & $5$ & $8$ & $9$ & $8$ & $5$ \\ \hline
    \end{tabular}
  \end{small}
\end{table}

From Table \ref{tab:1}, we see that the systematic SXOR code has smaller overheads and encoding complexity than the SXOR code for each $K$.
Moreover, we see that the systematic SXOR code has smaller overheads than the ZD code for each $K$.
However, the systematic SXOR code has larger encoding complexity than the ZD code for each $K$.
Summarizing above, we conclude that the systematic SXOR code has small overheads but its encoding complexity is high.

\section{Conclusion \label{sec:con}} 

In this paper, we have considered MAP decoding algorithm for the ZD code.
Moreover, we have proposed SXOR code and its systematic form which has small overheads under MAP decoding.
We see that the generator matrix can be classified into several classes.
Finally, we have evaluated the overhead of ZD codes and SXOR codes.
As a result, we have shown that the complexity and overhead are depended on the sequence $\bm{x}$ and primitive polynomial $g(z)$ and the systematic SXOR codes have smaller overhead than ZD codes.

\end{document}